\documentclass[10pt]{book}
\usepackage{array}
\usepackage{amsthm,amsmath,amssymb}
\usepackage{graphicx}
\usepackage{listings}
\RequirePackage{natbib}
\RequirePackage[colorlinks,citecolor=blue,urlcolor=blue]{hyperref}

\theoremstyle{plain}
\newtheorem{lemma}{Lemma}[section]
\newtheorem{theorem}{Theorem}[section]
\newtheorem{proposition}{Proposition}[section]

\newcommand{\bbx}{{\mathbf x}}
\newcommand{\al}{{\alpha}}
\newcommand{\te}{{\theta}}

\newcommand{\CC}{\mathbb{C}}
\newcommand{\la}{\lambda}\newcommand{\si}{\sigma}\newcommand{\Si}{\Sigma}
\newcommand{\us}{\underline{s}}
\newcommand{\de}{\delta}
\newcommand{\E}{\;  \mathbb{E}}

\begin{document}

\renewcommand{\thefootnote}{}
$\ $\par

\fontsize{10.95}{14pt plus.8pt minus .6pt}\selectfont
\vspace{0.8pc}
\centerline{\large\bf Estimation of the population spectral distribution from }
\vspace{2pt}
\centerline{\large\bf a large dimensional sample covariance matrix}
\vspace{.4cm}
\centerline{Weiming Li, Jiaqi Chen, Yingli Qin, Jianfeng Yao, Zhidong Bai}
\fontsize{9}{11.5pt plus.8pt minus .6pt}\selectfont

\footnote{
Weiming Li, School of Science, Beijing University of Posts and Telecommunications, Beijing, 100876, China. (E-mail: liwm601@gmail.com)}

\footnote{
Jiaqi Chen,
Department of Mathematics,
Harbin Institute of Technology,
Harbin, 150001,
China.
(E-mail: chenjq1016@gmail.com)}
\footnote{
Yingli Qin,
Statistics Division of Mathematical Sciences,
School of Physical and Mathematical Sciences,
Nanyang Technological University,
Singapore.
(E-mail: yingli.qin@gmail.com)}

\footnote{
Jianfeng Yao, Department of Statistics and Actuarial Sciences,
The University of Hong Kong, Hongkong, China. (E-mail: jeffyao@hku.hk)
}

\footnote{
Zhidong Bai,
KLAS MOE and School of Mathematics and Statistics,
Northeast Normal University,
Changchun, 130024.
China.
(E-mail: baizd@nenu.edu.cn)}

\begin{quotation}
\noindent {\it Abstract:}
This paper introduces a new method to estimate the spectral distribution of a population covariance matrix from high-dimensional data. The method is founded on a meaningful generalization of the seminal Mar\v{c}enko-Pastur equation, originally defined in the complex plan, to the real line. Beyond its easy implementation and the established  asymptotic consistency, the new estimator outperforms  two  existing estimators from the literature in almost all the  situations tested in a simulation experiment.
An application to the analysis of the correlation matrix of  S\&P stocks data is also given.
\par

\vspace{9pt}
\par
\end{quotation}\par

\section {Introduction}

 Let $\bbx_{1},\dots,\bbx_{n}$ be a sequence of i.i.d. zero-mean random vectors in
 $\mathbb{R}^{p}$ or $\CC^p$,
 with a common  population covariance  matrix
 $\Sigma_{p}$.
 When the population size $p$ is not negligible
 with respect to the sample size $n$,
 modern random matrix theory indicates that
 the sample covariance matrix
 \[    S_n = \frac1n \sum_{j=1}^n  \bbx_j  \bbx_j^*
 \]
 does not approach $\Sigma_p$.
For instance, in  a  simple case where $\Sigma_p=I_p$ (identity
matrix) , the eigenvalues of
$S_n$  will spread over an interval approximately equal to
$(1\mp\sqrt{p/n})^2$ around the unique
population eigenvalue 1 of $\Sigma_p$ (\cite{MP67}, \cite{YinB88} and \cite{BaiYin1993}). Therefore, classical statistical
 procedures based on an approximation of $\Sigma_p$ by $S_n$ become
 inconsistent  in  such  high dimensional data situations.

 To be precise, let us recall that the  {\em spectral   distribution} (SD) $G^A$
 of an  $m\times  m$ Hermitian matrix (or real symmetric)  $A$
 is  the measure generated by its eigenvalues
 $\{\la^A_i \}$,
 \[    G^{A} = \frac1m \sum_{i=1}^m \de_{\la^A_i}~,
 \]
 where $\de_b$ denotes the Dirac point measure at $b$.
 Let $(\si_{i})_{ 1\le i\le p}$ be the $p$ eigenvalues of the
 population covariance matrix $\Sigma_p$. We are particularly
 interested  in the following  SD
 \[  H_p := G^{\Si_p}= \frac{1}{p} \sum_{i=1}^p \delta_{\si_i}.
 \]
 Following  the random matrix theory, both sizes $p$ and
 $n$ will grow to infinity.  It is  then natural to assume that $H_p$ weakly
 converges to a limiting distribution $H$ when $p\to\infty$.
 We refer this  limiting SD  $H$ as the
 {\em population spectral distribution}  (PSD) of the observation
 model.

 The main observation is that under reasonable assumptions,
 when both dimensions $p$ and $n$ become large at a proportional rate
 say $c$,  almost surely,
 the (random) SD  $G^{S_n}$ of the sample covariance
 matrix $S_n$ will weakly converge  to a deterministic distribution
 $F$, called {\em limiting spectral distribution} (LSD).
Naturally this LSD $F$ depends on the PSD $H$, but in general
this relationship is  complex and has no explicit form.
The only exception is the case where all the population eigenvalues
 $(\si_i)$ are unit, i.e. $\Sigma_p\equiv I_p$ ($H=\de_1$);
the LSD $F$ is then
explicit known to be the Mar\v{c}enko-Pastur distribution with an
explicit density function. For a general PSD $H$, this relationship is
expressed via an  implicit equation, see Section \ref{sec:estim},
 Eqs. \eqref{equ2.1} and \eqref{equ2.3}.

 An  important question here is the  recovering of
 the PSD $H$ (or $H_p$) from the sample covariance matrix $S_n$.
 This question  has a central
 importance in several popular statistical methodologies like
 Principal Component Analysis (\cite{Johnstone01}),
 Kalman filtering or Independent
 Component Analysis which all rely on an efficient estimation of some
 population covariance matrices.

 Recently,  \cite{KarE08} has proposed a variational and
 nonparametric approach to this  problem
 based on  an appropriate distance function using the
 Mar\v{c}enko-Pastur equation
 \eqref{equ2.1} below  and a large dictionary  made with base density
 functions  and Dirac point masses.
 The proposed estimator is proved consistent in a nonparametric
 estimation sense  assuming both the  dictionary size and the number of observations $n$
 tend to infinity.  However,  no  result on the
 convergence rate of the estimator,  e.g. a
 central limit theorem, is given.

 In another important work \cite{RaoJ08}, the authors propose to use
 a suitable set of empirical moments, say the first $q$ moments: for $k=1,\ldots,q,$
$
   \hat{\alpha}_{k}=p^{-1}\mathop{\rm tr} S^k_n = p^{-1}
   \sum_{l=1}^{p}\lambda_{l}^{k}
$
 where   $(\la_l)$ are the eigenvalues of $S_n$ (assuming $p\le n$).
 Here a pure parametric approach is adopted and the PSD depends on a set of real parameters $\te$: $H=H(\te)$.
 Therefore,
 when $n\rightarrow\infty$ and under appropriate normalization,
 the sample moments  $(\hat{\alpha}_{k})$ will have a Gaussian
 limiting distribution
 with asymptotic mean and variance
 $\{m_{\theta},~Q_{\theta}\}$ which are functions of  the (unknown) parameters
 $\te$.
 In \cite{RaoJ08}, the authors propose an estimator
 $\hat{\theta}_R$  of the parameters by
 maximizing the asymptotic Gaussian likelihood of $\hat\al=(\hat\al_j)_{1\le j\le q}$,
with distribution $N_q(m_{\theta},Q_{\theta})$.
 Intensive simulations illustrate the consistency and the asymptotic normality of
 this estimator. However, their simulation experiments are limited to  simplest
 situations
 and
 no theoretic result are provided concerning the consistency of the
 estimator.
 An important   difficulty in this
 approach is that the functions $m_\theta$ and $Q_\theta$ have no explicit form.

 In a recent work \cite{Bai2010},
 a modification of the procedure
 in \cite{RaoJ08} is proposed to get  a direct moments estimator
 based on  the  sample moments $ (\hat{\alpha}_j)$.
 Compared to \cite{KarE08} and \cite{RaoJ08},
 this  moment estimator is  simpler and much easier to implement.
 Moreover, the convergence rate of this
 estimator (asymptotic normality) is also established.
 A recent paper by the authors in \cite{CBY11} has also
 analyzed the  underlying  order selection problem and
 proposed a  solution based on the   cross-validation principle.

 However, despite all the above  contributions,
 there is still
 a need for new methods of estimation.
 Actually,
 the  general approach in \cite{KarE08}
 has  several implementation issues that
 seem to be responsible
for its  relatively low performance
as attested by the very simple nature of
provided simulation results.
This low efficiency is probably due
to the use of a too general dictionary
made with  large number of
discrete distributions and piece-wisely linear densities.
Concerning the moment based methods in \cite{RaoJ08} and \cite{Bai2010},
we will see that their accuracy  degrades drastically as the number of parameters to be estimated increases. Lastly, it is well known that the contour-integral based method in a related work \cite{M08} is limited to a small class of discrete models where distinct population eigenvalues should generate non-overlapping clusters of sample eigenvalues.

The new approach developed in this paper can be viewed as a synthesis
of the optimization approach in  \cite{KarE08} and the parametric
setup in  \cite{Bai2010}. On one hand, we adopt  the
optimization approach  and will prove that it is in general preferable to the moment approaches.  On the other hand, using a generic parametric approach for
discrete PSDs as well as continuous PSDs, we are able to
avoid the aforementioned implementation difficulties in
\cite{KarE08}.
Another  important  contribution from the paper is that the
optimization problem has been moved from
the complex plan to the real line by
considering a characteristic equation (Mar\v{c}enko-Pastur equation)
on the real line.
The obtained optimization procedure is then much simpler than
the original one in \cite{KarE08}.

The rest of the paper is organised as follows.
In the next section, we provide a Mar\v{c}enko-Pastur equation
defined on the real line which will be the corner-stone of our
estimation method. This method is developed
in Section 3 and we prove  its strong consistency.
Then, in Section 4, simulation experiments are carried out
to compare the performance of three estimation methods under investigation.
The last section collects proofs of main
theorems.

\section{Mar\v{c}enko-Pastur equation on the real line}
\label{sec:MP}

Throughout the paper,
$A^{1/2}$ stands for any Hermitian square root of a non-negative definite
Hermitian matrix $A$. Our model assumptions are as follows.

\medskip
\noindent{\em Assumption}   (a). \quad
  The sample and population sizes $n,p$ both tend to infinity, and
  in such a way that $p/n\to c \in(0,\infty)$.

\medskip
\noindent{\em Assumption}   (b). \quad
  There is a doubly
  infinite array of i.i.d. complex-valued random variables
  $(w_{ij})$, $i,j\ge 1$
  satisfying
  \[  \label{eq:moments}
  \E(w_{11})=0,~~
  \E(|w_{11}|^2)=1,
  \]
  such that   for each $p,n$,  letting $W_n=(w_{ij})_{1\le i\le p,1\le j \le n}$,
  the observation vectors
  can be represented as
  $\bbx_j=\Si_p^{1/2}{w_{.j}}$ where $w_{.j}=(w_{ij})_{1\le i\le p}$
  denotes the $j$-th column of $W_n$.

\medskip
\noindent{\em Assumption}   (c). \quad
  The SD  $H_p$ of    $\Sigma_{p}$ weakly converges
  to a probability distribution  $H$ as
  $n\to\infty$.

\medskip

The assumptions (a)-(c) are classical conditions
for  the celebrated Mar\v{c}enko-Pastur theorem (\cite{MP67,Silverstein95}, see
also \cite{BSbook}). More precisely,
under these  Assumptions, almost
surely, as $n\to \infty$, the empirical SD $ F_n:=G^{{S}_{n}}$  of  ${S}_{n}$,
weakly converges to a (nonrandom)
generalized Mar\v{c}enko-Pastur  distribution
$F$.

Unfortunately,  except the simplest case
where $H\equiv \de_1$, the LSD
$F$ has no explicit form and it is characterized as  follows.
Let $\us(z)$ denote  the Stieltjes transform of
$
  cF + (1-c)\de_0~,
$
which is a one-to-one map defined on the  upper half
complex plan $\CC^+=\{ z\in \CC:~ \Im(z) >0\}$.
This transform  satisfies the following
fundamental  Mar\v{c}enko-Pastur  equation (MP):
\begin{equation}  \label{equ2.1}
  z  =  - \frac1 {\us(z)}  +  c \int\!\frac{t}{1+t\us(z)} dH(t)~,\quad z\in\CC^+.
\end{equation}

The above MP  equation excludes the real line from its domain of
definition.
As the  first contribution of  the paper, we fill this gap by an
extension of the MP equation to the real line.
The estimation method introduced in Section 3 will be entirely
based on this extension.

The support of a
distribution $G$ is denoted by  $S_G$
and its complementary set by $S_G^c$, since the ESD $F_n$ is observed, we will
use $\us_n$, the Stieltjes transform of
$(p/n)F_n + (1-p/n)\de_0$ to approximate $\us$ in the MP equation. More precisely, let for $u\in\mathbb R$,
\begin{equation}\label{equ2.2}
\us_n(u)=-\frac{1-p/n}{u}+\frac{1}{n}\sum_{l=1}^p\frac{1}{\lambda_l-u}.
\end{equation}
It is clear
that the domain of
$\underline{s}_n(u)$ is $S_{F_n}^c$. Thus, $\underline{s}_n(u)$'s are
well defined on $\mathring{U}$ for all large $n$, where
$\mathring{U}$ is the interior of
$U=\liminf_{n\rightarrow\infty}S_{F_n}^c\setminus\{0\}$.

\begin{theorem}\label{th1}
  Assume that the assumptions (a)-(b)-(c) hold.
  Then
\begin{itemize}
\item [{\rm (1)}] for any $u\in \mathring{U}$, $\underline{s}_n(u)$ converges to $\underline{s}(u)$,
\item [{\rm (2)}] for any $u\in S_F^c$, $\underline{s}=\underline{s}(u)$ is a solution to equation
\begin{eqnarray}\label{equ2.3}
u=-\frac{1}{\underline{s}}+c\int\frac{t}{1+t\underline{s}}dH(t),
\end{eqnarray}
\item [{\rm (3)}] the solution is also unique in the set $B^+=\{\underline{s}\in\mathbb R\backslash\{0\}: du/d\underline{s}>0, (-\underline{s})^{-1}\in S_H^c\}$,
\item [{\rm (4)}]  for any non-empty open interval $(a, b)\subset B^+$,
$H$ is uniquely determined by $u(\us),\ \us\in(a, b)$.
\end{itemize}
\end{theorem}

The proof is given in the last section. Some remarks are in order.
\begin{enumerate}
\item
  Notice that since $(-\infty,0)\subset\mathring{U}\subset S_F^c$, there are infinitely many $u$-points such that
  $\underline{s}_n(u)$ almost
  surely converges to $\underline{s}(u)$.
\item
  The MP equation \eqref{equ2.3} can be inverted in the following
  sense:
  the knowledge of $u(\underline s)$ on any interval in
  $B^+$ (see Figure \ref{fig1}) will uniquely determine the PSD $H$.
  The estimation method in Section~3 will be built
  on this property.
\end{enumerate}

\begin{figure}[th]
\centering
\includegraphics[width=3in,height=2in]{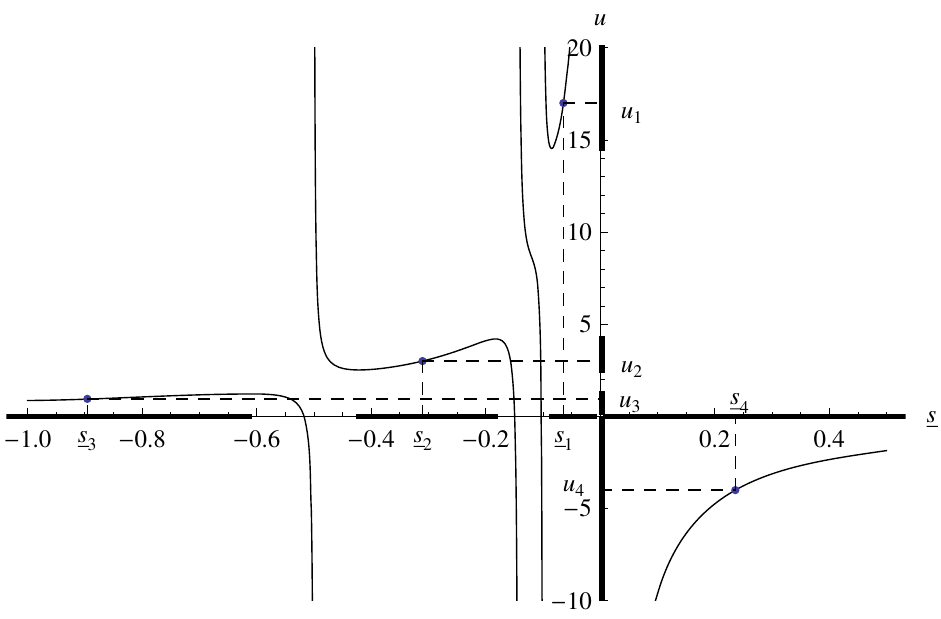}
\caption{The curve of $u=u(\underline{s})$ (solid thin), and the sets $B^+$ and $S_{F}^c$ (solid thick) for $H=0.3\delta_2+0.4\delta_7+0.3\delta_{10}$ and $c=0.1$. $u_i=u(\underline{s}_i), \underline{s}_i\in B^+, i=1,2,3,4$.
}
\label{fig1}
\end{figure}

\section {Estimation}
\label{sec:estim}
\subsection{The method}

We consider the estimation problem in a parametric setup.
Suppose $H=H(\theta)$ is the limit of $H_p$ with unknown parameter vector $\theta\in\Theta\subset{\mathbb R^q}$. The procedure of the estimation of $H$ includes three steps:
\begin{itemize}
\item [{\rm S1.}] Choose a $u$-net $\{u_1,\ldots, u_m\}$ from $\mathring{U}$, where $u_j$'s are distinct and the size $m$ is no less than $q$.
\item [{\rm S2.}]
For each $u_j$, calculate $\underline{s}_n(u_j)$ using \eqref{equ2.2} and plug the pair into the MP equation \eqref{equ2.3}. Then, we obtain $m$ approximate equations
    \begin{eqnarray*}u_j&\simeq&-\frac{1}{\underline{s}_n(u_j)}+\frac{p}{n}\int\frac{tdH(t,\theta)}{1+t\underline{s}_n(u_j)}\\
    &:=&\widehat{u}_j(\underline{s}_{nj},\theta) \quad (j=1,\ldots,m).
    \end{eqnarray*}
\item [{\rm S3.}]
Find the least squares solution of $\theta$,
    \begin{eqnarray*}
    \widehat{\theta}_n&=&\arg\min_{\theta\in\Theta}\sum_{j=1}^{m}\bigg(u_j-\widehat{u}_j(\underline{s}_{nj},\theta)\bigg)^2.
    \end{eqnarray*}
\end{itemize}

We name $\widehat{\theta}_n$ as the least squares estimate (LSE) of $\theta$. Accordingly, $\widehat{H}=H(\widehat{\theta}_n)$ is called the LSE of $H$. A central issue here is the choice of the $u$-net $\{u_1,\ldots,u_m\}$. In Section 4, we will provide a robust method for this choice that can be used in practice with real data.

This procedure can also be applied to the MP equation \eqref{equ2.1} in complex field as in \cite{KarE08}. Similarly  to our first two steps, the author chose a $z$-net from $\mathbb C^+$ and created a system of approximate equations by a discretisation $H$ as a weighted sum of a grid of pre-chosen mass points. The estimates of the weight parameters were then obtained by minimizing the approximation errors in terms of the $L_\infty$ norm. The author also suggested to use a $z$-net with $\Re{(z)}<0$ and $\Im{(z)}$ near $0$. This is almost equivalent to choosing a $u$-net with $u<0$ in our procedure. But we strongly suggest to use more $u$-points from $\mathring{U}\cap\mathbb R^+$ if possible, since these points are likely to carry some different information about $H$ comparing with negative $u$-points. For the optimization step, whatever the distance used ($L_2$-norm, $L_\infty$-norm, etc.) our method would be easier and faster than El Karoui's one since the optimization is carried on the real domain.

\subsection{Consistency}

We establish the strong consistency of our estimator in two models that are widely used in the literature. The estimates will be further studied in the simulation section.

The first model is made with discrete PSDs with finite support on $\mathbb R^+$, i.e.
$$H(\theta)=m_1\delta_{a_1}+\cdots+ m_k\delta_{a_k},\quad \theta\in\Theta,$$ where $m_k=1-\sum_{i=1}^{k-1}m_i$, $\theta=(a_1,\ldots, a_k, m_1,\ldots, m_{k-1})$ are ($2k-1$) unknown parameters and
\begin{eqnarray*}
\Theta =\bigg\{\theta\in\mathbb R^{2k-1}: m_i>0, \sum_{i=1}^{k}m_i=1;  0< a_1<\cdots <a_k<+\infty \bigg\}.
\end{eqnarray*}
Here, Equation \eqref{equ2.3} can be simplified to
\begin{eqnarray*}
u=-\frac{1}{\underline{s}}+c\sum_{i=1}^{k}\frac{a_im_i}{1+a_i\underline{s}}.
\end{eqnarray*}
For the well-definition of the equation on $\Theta$, we assume that the $u$-net satisfies
\begin{eqnarray}\label{equ3.1}
\inf_{\theta\in\Theta}\min_{i, j}|1+a_i\underline{s}(u_j)|\geq\delta,
\end{eqnarray}
where $\delta$ is some positive constant. It is clearly satisfied if all the $u_j$'s are negative.

\begin{theorem}\label{th2}
In addition to the assumptions (a)-(b)-(c), suppose that the true value of the parameter $\theta_0$ is an inner point of $\Theta$ and the condition \eqref{equ3.1} is fulfilled. Then, the LSE $\widehat{\theta}_n$ for the discrete model is strongly consistent, that is, almost surely,
$\widehat{\theta}_n\rightarrow \theta_0.$
\end{theorem}

Next we suppose that the PSD $H(\theta)$ has a probability density $h(t|\theta)$ with respect to Lebesgue measure. From \cite{Sbook} (Chapters 2, 4), if $h(t|\theta)$ has finite moments of all order, it can be expanded in terms of Laguerre polynomials:
$$h(t|\theta)=\sum_{j\geq0} c_j\psi_j(t)e^{-t},$$
where $$c_j=\int\psi_j(t)h(t|\theta)dt.$$
As discussed in \cite{Bai2010}, we consider a family of $h(t|\theta)$ with finite expansion $$h(t|\theta)=\sum_{j=0}^q c_j\psi_j(t)e^{-t}=\sum_{j=0}^{q} \alpha_jt^je^{-t},\quad t>0,\quad \theta\in\Theta,$$
where $\alpha_0=1-\alpha_1-\cdots-q!\alpha_q$, $\theta=(\alpha_1,\ldots,\alpha_q)$, and
$$\Theta=\big\{\theta\in \mathbb R^q: h(t|\theta)>0,\ t\in\mathbb R^+ \big\}.$$
For this model, Equation \eqref{equ2.3} becomes
\begin{eqnarray*}
u=-\frac{1}{\underline{s}}+c\sum_{j=0}^{q}\alpha_j\int\frac{t^{j+1}e^{-t}}{1+t\underline{s}}dt.
\end{eqnarray*}
It's clear that the calculation of $\widehat{\theta}_n$ is here simple since the above equation is linear with respect to $\theta$.

\begin{theorem}\label{th3}
In addition to the assumptions (a)-(b)-(c), suppose that the true value of the parameter $\theta_0$ is an inner point of $\Theta$. Then, the LSE $\widehat{\theta}_n$ for the continuous model is strongly consistent.
\end{theorem}

\section{Simulation experiments}

In this section, simulations are carried out to compare our LSE with the approximate quasi-likelihood estimate in \cite{RaoJ08} (referred as RMSE) and the moment estimate in \cite{Bai2010} (referred as BCY). We do not include the estimator of \cite{KarE08} in this study since this estimator is nonparametric using a suitable approximation dictionary while the LSE is based on a parametric form of unknown PSDs.

We study five different PSDs: three of them are discrete and two continuous. Samples are drawn from mean-zero real normal population with the dimensions $n=500$ and $p=100, 500, 1000$. Statistics are computed from $1000$ independent replications.

To evaluate the quality of an estimate $\widehat{H}=H(\widehat{\theta})$, instead of looking at individual values ($\widehat{\theta}_i$) of the parameters, we use a global distance, namely the Wasserstein distance $W=\int |Q_{H}(t)-Q_{\widehat{H}}(t)|dt$ where $Q_{\mu}(t)$ is the quantile  function of distribution $\mu$. The use of Wasserstein distance is motivated by the fact that it applies to both discrete and continuous distributions (unlike other common distance like kullback-leibler or $L_2$ distance).

For the LSE, we need to choose a $u$-net from $S_{F_n}^c\cap S_{F}^c\setminus\{0\}$. When $H$ has finite support, the upper and lower bounds of $S_{F}\setminus\{0\}$ can be estimated respectively by $\lambda_{max}=\max\{\lambda_i\}$ and $\lambda_{min}=\min\{\lambda_i: \lambda_i>0\}$ where $\lambda_i$'s are sample eigenvalues. As a consequence, we design a primary set:
$$
\mathcal U=
\begin{cases}
(-10, 0)\cup(0, 0.5\lambda_{min})\cup(5\lambda_{max}, 10\lambda_{max})
& (\text{discrete model},\ p\neq n),\\
(-10, 0)\cup(5\lambda_{max}, 10\lambda_{max})
& (\text{discrete model},\ p=n),\\
(-10, 0)& (\text{continuous model}).
\end{cases}
$$
Next, we choose $l$ equally spaced  $u$-points from each individual interval of $\mathcal U$. We name this process as adaptive choice of $u$-net. Here we set $l=20$ for all cases considered in simulation, that is, for example we take $\{-10+10t/21,t=1,\ldots,20\}$ from the first interval.

{\bf Case 1: $H=0.5\delta_1+0.5\delta_2$.} This is a simple case as $H$ has only two atoms with equal weights. Table \ref{table1} shows that all the three estimates are consistent, and their efficiency is very close.

\begin{table}
\begin{center}
\caption{Wasserstein distances of estimates for $H=0.5\delta_1+0.5\delta_2$.}
\begin{tabular}{llccccccccc}
\hline
&             &$p/n=0.2$&$p/n=1$&$p/n=2$\\
\hline
LSE      &Mean&0.0437&0.0601&0.0893\\
         &S.D.&0.0573&0.0735&0.1077\\
RMSE     &Mean&0.0491&0.0689&0.0859\\
         &S.D.&0.0320&0.0482&0.0629\\
BCY      &Mean&0.0500&0.0664&0.0871\\
         &S.D.&0.0331&0.0466&0.0617\\
\hline
\end{tabular}
\label{table1}
\end{center}
\end{table}

{\bf Case 2: $H=0.3\delta_1+0.4\delta_3+0.3\delta_5$.} In this case, we increase the order of $H$. Analogous statistics are summarized in Table \ref{table2}. The results show that LSE clearly outperforms RMSE and BCY in the light of the Wasserstein distance. Particularly, RMSE and BCY have not converged yet with dimensions $n=500$ and $p=500, 1000$, while LSE only contains a small bias in such situations. This exhibits the robustness of our method with respect to the increase of the order.

\begin{table}
\begin{center}
\caption{Wasserstein distances of estimates for $H=0.3\delta_1+0.4\delta_3+0.3\delta_5$.}
\begin{tabular}{llccccccccc}
\hline
&             &$p/n=0.2$&$p/n=1$&$p/n=2$\\
\hline
LSE      &Mean&0.1589&0.3566&0.4645\\
         &S.D.&0.1836&0.4044&0.5156\\
RMSE     &Mean&0.2893&0.7494&0.8153\\
         &S.D.&0.0966&0.2188&0.1080\\
BCY      &Mean&0.2824&0.5840&0.7217\\
         &S.D.&0.1769&0.2494&0.2156\\
\hline
\end{tabular}
\label{table2}
\end{center}
\end{table}

{\bf Case 3: $H=0.3\delta_1+0.4\delta_5+0.3\delta_{15}$.} In this case, we increase the variance of $H$. Table \ref{table3} collects the simulation results. Compared with Table \ref{table2}, RMSE and BCY deteriorate significantly while LSE remains stable. The average Wasserstein distances of LSE are (at least) a third less than those of RMSE and BCY for all $p$ and $n$ used. This demonstrates the robustness of our method with respect to the increase of the variance.

\begin{table}
\begin{center}
\caption{Wasserstein distances of estimates for $H=0.3\delta_1+0.4\delta_5+0.3\delta_{15}$.}
\begin{tabular}{llccccccccc}
\hline
&             &$p/n=0.2$&$p/n=1$&$p/n=2$\\
\hline
LSE      &Mean&0.1756&0.2524&0.5369\\
         &S.D.&0.2105&0.3013&0.6282\\
RMSE     &Mean&0.7090&1.4020&1.9160\\
         &S.D.&0.0524&0.6501&0.2973\\
BCY      &Mean&0.9926&1.5379&1.8562\\
         &S.D.&0.5618&0.6875&0.7526\\
\hline
\end{tabular}
\label{table3}
\end{center}
\end{table}

{\bf Case 4: $h(t)=(\alpha_0+\alpha_1t)e^{-t}, \alpha_1=1$.} This is the simplest continuous model with only one parameter to be estimated. In this case, $H$ is a gamma distribution with shape parameter 2 and scale parameter 1. Statistics in Table \ref{table4} show that all the three estimates have similar efficiency.

\begin{table}
\begin{center}
\caption{Wasserstein distances of estimates for $h(t)=te^{-t}$.}
\begin{tabular}{llccccccccc}
\hline
&             &$p/n=0.2$&$p/n=1$&$p/n=2$\\
\hline
LSE      &Mean&0.0939&0.0441&0.0294\\
         &S.D.&0.0704&0.0317&0.0229\\
RMSE     &Mean&0.1126&0.0508&0.0346\\
         &S.D.&0.0839&0.0393&0.0262\\
BCY      &Mean&0.1168&0.0491&0.0348\\
         &S.D.&0.0881&0.0361&0.0268\\
\hline
\end{tabular}
\label{table4}
\end{center}
\end{table}

{\bf Case 5: $h(t)=(\alpha_0+\alpha_1t+\alpha_2t^2+\alpha_3t^3)e^{-t}, \alpha_1=\alpha_2=\alpha_3=1/9$.} This model with three parameters becomes more difficult to estimate. RMSE and BCY have large bias and/or large standard deviations in all dimensions we used, see Table \ref{table5}. In contrast, our LSE performs fairly well and again outperform these two moment based methods.

\begin{table}
\begin{center}
\caption{Wasserstein distances of estimates for $h(t)=(t+t^2+t^3)e^{-t}/9$.}
\begin{tabular}{llccccccccc}
\hline
&             &$p/n=0.2$&$p/n=1$&$p/n=2$\\
\hline
LSE      &Mean&0.1895&0.0902&0.0740\\
         &S.D.&0.1103&0.0526&0.0378\\
RMSE     &Mean&0.3163&0.1515&0.1156\\
         &S.D.&0.2062&0.0863&0.0670\\
BCY      &Mean&0.3139&0.1554&0.1114\\
         &S.D.&0.2007&0.0907&0.0624\\
\hline
\end{tabular}
\label{table5}
\end{center}
\end{table}

In summary, the LSE outperforms the RMSE and BCY estimators in all the tested situations. On the other hand, as expected, the performances of the RMSE and the BCY estimators are very close since they are all based on empirical moments (however, as explained in \cite{Bai2010}, the BCY estimator is much easier to implement).

Finally, we analyze the relationship between the size of a $u$-net and the efficiency of LSE. The average of Wasserstein distances of LSE with respect to different $l$ values (the number of $u$-points picked from each individual interval) is plotted for Case 3 and Case 5, see Figure \ref{fig2}. The results show that unless $l$ is too small, the estimation efficiency remains remarkably stable with different values of $l$.

\begin{figure}[h]
\begin{minipage}[t]{0.5\linewidth}
\includegraphics[width=2.5in,height=1.8in]{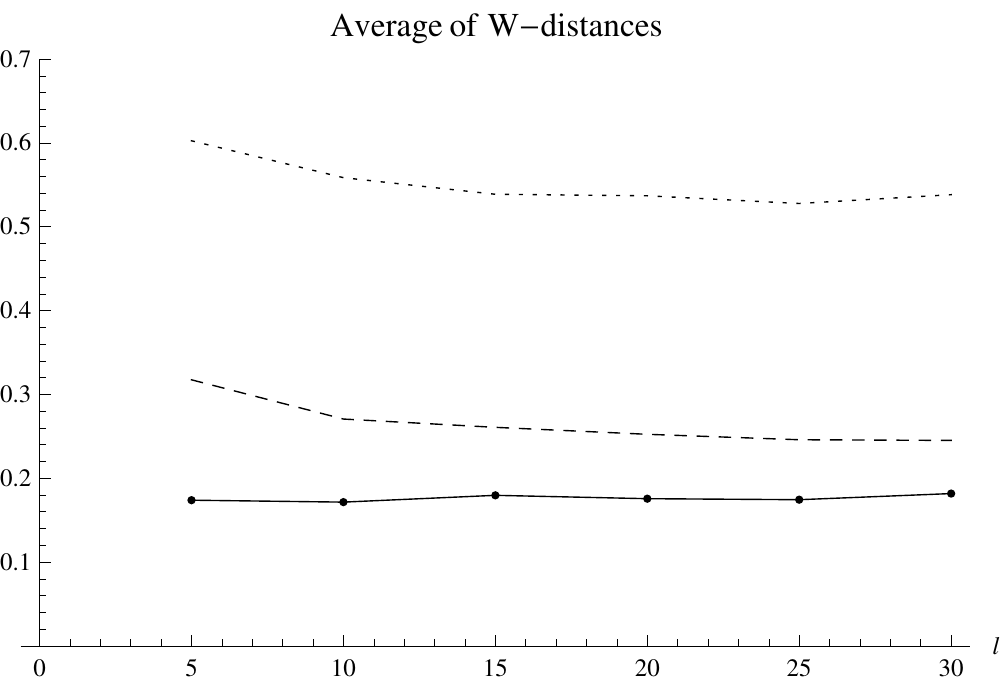}
\end{minipage}%
\begin{minipage}[t]{0.5\linewidth}
\includegraphics[width=2.5in,height=1.8in]{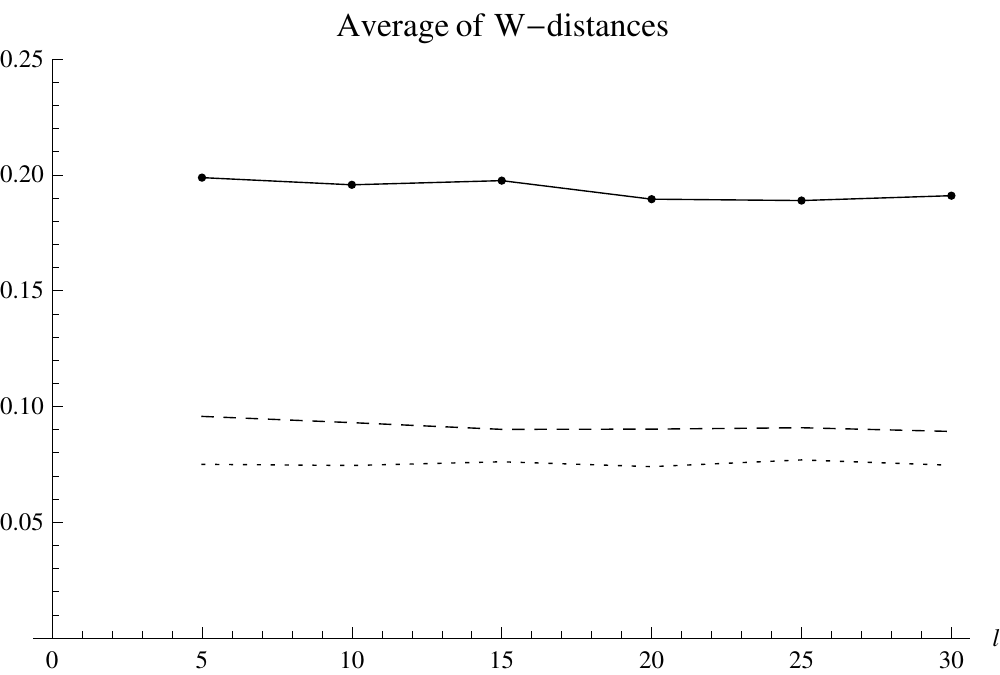}
\end{minipage}
\caption{The average of Wasserstein distances of LSE with respect to $l$ ($l=5,10,\ldots,30$) for Case 3 (left) and Case 5 (right) with $p=100, n=500$ (solid lines), $p=500, n=500$ (dashed lines), and $p=1000, n=500$ (dotted lines).}
\label{fig2}
\end{figure}

\section{Application to S\&P 500 stocks data}

In this section, we present a financial application of our estimation procedure in analysing an empirical correlation matrix of stock returns. We study a set of 488 U.S. stocks included in the S\&P 500 index from September, 2007 to September 2011 (1001 trading days, 12 stocks have been removed because of missing values). Here, the data dimension is $p=488$ and the number of observations is $n=1000$.

Following \cite{BP09}, we suppose that there is a PSD $H(\alpha)$ for the stock returns with an inverse cubic density $h(t|\alpha)$:
$$h(t| \alpha)=\frac{c}{(t-a)^3}I(t\geq\alpha),\quad 0\leq \alpha< 1,$$
where $c=2(1-\alpha)^2$ and $a=2\alpha-1$. Notice that when $\alpha\rightarrow1^-$, the inverse cubic model tends to the MP case ($H=\delta_1$), so that this prior model is very flexible.

For the estimation procedure, we first remove the 6 largest sample eigenvalues which are deemed as spikes over the bulk of sample eigenvalues. As in Section 3, we use $l=20$ equally spaced $u$-points in $(-10, 0)$. The LSE of $\alpha$ turns out to be $\widehat{\alpha}=0.4380$. The RMSE and BCY don't exist for this model for the reason that the moments of $H$ don't depend on the unknown parameter.

Limiting spectral densities corresponding to the LSE estimate $h(t|0.4380)$ and $H=\delta_1$ are shown in Figure \ref{fig3}. We also plot the empirical spectral density of the correlation matrix, and the curve is smoothed by using a Gaussian kernel estimate with bandwidth $h=0.05$.

\begin{figure}[h]
\begin{center}
\includegraphics[width=2.5in,height=1.8in]{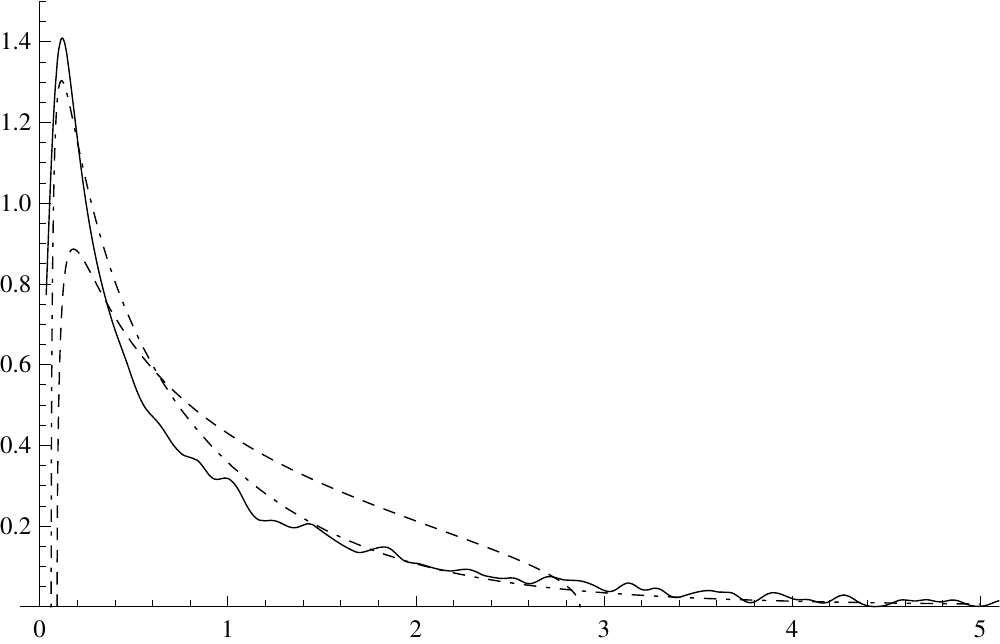}
\caption{The empirical density of the sample eigenvalues (plain black line), compared to the MP density (dashed line) and the limiting spectral density corresponding to the LSE estimate $h(t|0.4380)$ (dashed-dotted line). }
\label{fig3}
\end{center}
\end{figure}

From Figure \ref{fig3}, we could see that the MP density is far away from the empirical density curve. This confirms a widely believed fact that the correlation matrix may have more structure than just several spikes on top of the identity matrix. By contrast, the cubic model with $\alpha=0.4380$ yields a much more satisfying fit to the empirical density curve.

\section {Proofs}

We first recall useful results in three lemmas.  The first one is
provided  in \cite{Silverstein95} and the two others
in \cite{SilversteinC95}.

\begin{lemma}\label{lema1}
Assume that the assumptions (a)-(b)-(c) hold. Then, almost surely, the empirical spectral distribution $F_n$ converges in distribution, as $n\rightarrow\infty,$ to a non-random probability measure $F$,
whose Stieltjes transform $s=s(z)$ is a solution to the equation
\begin{eqnarray*}
s=\int\frac{1}{t(1-c-czs)-z}dH(t).
\end{eqnarray*}
The solution is also unique in the set $\{s\in {\mathbb C}: -(1-c)/z+cs\in{\mathbb C^+}\}.$
\end{lemma}

\begin{lemma}\label{lema2}
If $u\in S_{F}^c\setminus\{0\}$, then $\underline{s}=\underline{s}(u)$ satisfies
\begin{align*}
\text{\rm (1)}\ & \underline{s}\in \mathbb{R}\setminus\{0\}, &
\text{\rm (2)}\ & (-\underline{s})^{-1}\in S_H^c, &
\text{\rm (3)}\ du/d\underline{s}>0.
\end{align*}
Conversely, if $\underline{s}$ satisfies {\rm(1)-(3)}, then $u=u(\underline{s})\in S_{F}^c\setminus\{0\}$.
\end{lemma}

\begin{lemma}\label{lema3}
Set $B=\{\underline{s}\in{\mathbb R}\backslash\{0\}:(-\underline{s})^{-1}\in S_H^c\}$. Let $[\underline{s}_1, \underline{s}_2]$, $[\underline{s}_3, \underline{s}_4]$ be two disjoint intervals in $B$ satisfying for all $\underline{s}\in [\underline{s}_1, \underline{s}_2]\cup[\underline{s}_3, \underline{s}_4]$, $du/d\underline{s}>0$. Then $[u_1, u_2], [u_3, u_4]$ are disjoint where $u_i = u(\underline{s}_i),
i = 1, 2, 3, 4$.
\end{lemma}

\subsection{Proof of Theorem \ref{th1}}\label{th1proof}
The first conclusion follows from two convergence theorem. In fact, for any fixed $u\in \mathring{U}$ there are $\varepsilon_0\in \mathbb R^+$ and $n_0\in \mathbb Z^+$ such that $U(u,\varepsilon_0)\subset\cap_{n=n_0}^{\infty}S_{F_n}^c\setminus\{0\}.$ This implies that for all $x\in S_{F_n}$ and $n>n_0$, we have $|1/(x-u)|<1/\varepsilon_0.$ From this and Lebesgue's dominated convergence theorem, for any fixed $u\in \mathring{U}$, almost surely,
\begin{eqnarray*}
\underline{s}_n(u)\rightarrow\underline{s}(u),
\end{eqnarray*}
as $n\rightarrow\infty$ with $p/n\rightarrow c>0$. By Vitali's convergence theorem
\citep{Tbook}, we may conclude that $\underline{s}_n(u)$ converges almost surely for every $u\in \mathring{U}$.

Next, we consider the second conclusion. For any fixed $u\in S_F^c$ and $\varepsilon>0$, let $z=u+\varepsilon i$, from Lemma \ref{lema1}, $\underline{s}(z)$ satisfies \eqref{equ2.1}. On the other hand, according to Lemma \ref{lema2}, $(-\underline{s}(u))^{-1}\in S_H^c$ and thus $|t/(1+t\underline{s}(z))|$ is bounded on the set $\{(t,\varepsilon):S_H\times[0,1]\}$. Therefore, by Lebesgue's dominated convergence theorem, taking the limit as $\varepsilon\rightarrow 0^{+}$ on both sides of \eqref{equ2.1} gets the conclusion.

Conclusion 3 follows from the results of Lemma \ref{lema2} and Lemma \ref{lema3}. In fact, let $u_{B^+}(\underline{s})$ be the restriction of $u(\underline{s})$ to $B^+$, then Lemma \ref{lema2} shows that the range of $u_{B^+}(\underline{s})$ is $S_{F}^c\setminus\{0\}$. Lemma \ref{lema3} indicates that $u_{B^+}(\underline{s})$ is also an injection. Therefore, $u_{B^+}(\underline{s})$ is a bijection from $B^+$ to $S_{F}^c\setminus\{0\}$.

As to the last conclusion, suppose $H_1(t)$ and $H_2(t)$ are two population spectral distribution functions satisfying, for all $\underline{s}\in (a, b)$,
\begin{eqnarray}
\int\frac{t}{1+t\underline{s}}dH_1(t)=\int\frac{t}{1+t\underline{s}}dH_2(t).\label{equ7.2}
\end{eqnarray}
We are going to show $H_1=H_2$ almost everywhere with respect to Lebesgue measure on $\mathbb R$.

For any $\underline{s}_0\in (a, b)$, $-1/\underline{s}_0$ is an inner point of $S_H^c$, then there is $\delta_0>0$ such that $$U(-1/\underline{s}_0, \delta_0/|\underline{s}_0|)\subset S_H^c,$$ which implies $|1+t\underline{s}_0|>\delta_0$ for all $t\in S_H$. Choose $\varepsilon_0=\min\{|\underline{s}_0|\delta_0/(1+\delta_0),
b-\underline{s}_0, \underline{s}_0-a\}$. Then, for any $\underline{s}\in U(\underline{s}_0,\varepsilon_0)$, $1+t\underline{s}$ has the same sign as $1+t\underline{s}_0$. Define
$$
g(t,u)=
\begin{cases}
1 &(1+t\underline{s}_0>0,\ u>0),\cr
-1&(1+t\underline{s}_0<0,\ u<0),\cr
0 &(u(1+t\underline{s}_0)u\leq0).
\end{cases}
$$
We have then
\begin{eqnarray*}
\frac{t}{1+t\underline{s}}&=&\int_{-\infty}^{+\infty}
g(t,u)e^{-(1/t+\underline{s})u}du,\ \underline{s}\in
U(\underline{s}_0,\varepsilon_0).
\end{eqnarray*}
Therefore, each side of \eqref{equ7.2} can be expressed as
\begin{eqnarray}
\int\frac{t}{1+t\underline{s}}dH_i(t)=\int_{-\infty}^{+\infty}\int_{0}^{\infty}g(t,u)
e^{-(\frac{1}{t}+\underline{s}_0)u}dH_i(t)e^{-(\underline{s}-\underline{s}_0)u}du.\label{equ7.3}
\end{eqnarray}

It is clear that the left hand side of \eqref{equ7.3} is the Laplace transform of $$\int_{0}^{\infty}g(t,u)
e^{-(\frac{1}{t}+\underline{s}_0)u}dH_i(t).$$ By the uniqueness of Laplace transform, we have then
\begin{eqnarray*}
\int_{0}^{\infty}g(t,u)e^{-\frac{1}{t}u}dH_1(t)=\int_{0}^{\infty}g(t,u)e^{-\frac{1}{t}u}dH_2(t),
\end{eqnarray*}
and thus $H_1=H_2$ almost everywhere.

\subsection{ Proof of Theorem \ref{th2}}\label{th2proof}
Define
\begin{equation}
\varphi(\theta)=\sum_{j=1}^m\bigg(u_j-u(s_j,\theta)\bigg)^2,\label{equ7.6}\nonumber
\end{equation}
where $s_j=\underline{s}(u_j)\ (j=1,\ldots,m)$. We first state and prove the following proposition.

\begin{proposition}\label{pro2}
If $u_1,\ldots,u_m$ are distinct and $m\geq q=2k-1$, then $\varphi(\theta)=0$ for the discrete model has a unique solution $\theta_0$ on $\Theta$.
\end{proposition}

\begin{proof}
Since $\underline{s}(u)$ is a bijective function from $S_F^c$ to $B^+$ and $u_1, \ldots, u_m$ are distinct, $s_1, \ldots, s_m$ are also distinct.

Suppose there is a $\theta=(a_1,\ldots, a_k, m_1,\ldots,m_{k-1})$ such that $\varphi(\theta)=0$.
Denote by $\theta_0=( a'_1,\ldots, a'_k, m'_1, \ldots, m'_{k-1} )$ the true value of the parameter. We will show that $\theta=\theta_0$.
Denote $b_i=1/a_i$ and $b'_i=1/a'_i\ (i=1,\ldots,k)$, we have then
\begin{equation}\label{equ7.6-b}
   \sum_{i=1}^k \frac{m_i}{s_j+b_i} =   \sum_{i=1}^k
   \frac{m_i'}{s_j+b_i'}  \quad (j=1,\ldots,m).
\end{equation}
Now look $s_j$ as a parameter $s$ and reduction to common factors
leads to
\[
(s+b_1')\cdots(s+b_k') \sum_{i=1}^k m_i \prod_{\ell\ne i}
(s+b_\ell)  =
(s+b_1)\cdots(s+b_k) \sum_{i=1}^k m_i' \prod_{\ell\ne i}
(s+b_\ell') .
\]
These are polynomials of degree $2k-1$; they coincide at  $m\ge 2k-1$
different points $s=s_j$;   they are then equal.
Back to \eqref{equ7.6}, we have now for all $s\neq -b_i, -b_i'$,
\[      \sum_{i=1}^k \frac{m_i}{s+b_i} =   \sum_{i=1}^k
   \frac{m_i'}{s+b_i'}.
\]
Now each $b_i$ should match one $b_\ell'$, because otherwise $b_i\ne
b_\ell'$ for all $\ell$ and by  letting $ s\to -b_i$ we get a contradiction.
So there is one $b_\ell'$ matches (then unique) for $b_i$.
This proves also that $m_i=m_\ell'$. As the $b_i$ are ordered,
it is necessary that $b_\ell'=b_i'$ and hence also $m_i=m_i'$.
\end{proof}

Now let's begin the proof of Theorem \ref{th2}.
Recall that
\begin{eqnarray*}
\widehat{\theta}_n&=&\arg\min_{\theta\in\Theta} \sum_{j=1}^{m}\bigg(u_j+\frac{1}{\underline{s}_n(u_j)}-\frac{p}{n}\int\frac{tdH(t,\theta)}{1+t\underline{s}_n(u_j)}\bigg)^2\\
&:=&\arg\min_{\theta\in\Theta}\varphi_n(\theta).
\end{eqnarray*}
Under the assumption of the theorem, by the convergence of $\underline{s}_n(u_j)\ (j=1,\ldots,m)$, $\varphi_n(\theta)$ is well defined on $\Theta$ for all large $n$. Moreover, for any fixed $\theta\in\Theta$, we have
\begin{equation}
\varphi_n(\theta)\rightarrow\varphi(\theta),\label{equ7.7}\nonumber
\end{equation}
almost surely.
Proposition \ref{pro2} guarantees that $\theta=\theta_0$ is the unique solution to $\varphi(\theta)=0$ on $\Theta$.

We claim that for almost all $\omega$, there is a compact set $\overline{\Theta}=\overline{\Theta}(\omega)\subset{\Theta}$ which contains all $\widehat{\theta}_n(\omega)$ for large $n$.
It's easy to see that for all large $n$, $\varphi_n(\theta)$ is uniformly bounded on $\overline{\Theta}$ and has continues partial derivatives with respect to $\theta$. By the Vitali's convergence theorem, we get
\begin{equation}\label{equ7.8}
\sup_{\theta\in\overline\Theta}|\varphi_n(\theta)-\varphi(\theta)|\rightarrow0.
\end{equation}
For any $\varepsilon>0$, by the continuity of
$\varphi(\theta)$, we have
$$
\inf_{||\theta-\theta_{0}||>\varepsilon\atop \theta\in\overline\Theta}
\varphi(\theta)>\varphi(\theta_0)=0.
$$
From this and \eqref{equ7.8}, when $n$ is large,
\begin{equation}\nonumber
\inf_{\|\theta-\theta_0\|>\varepsilon\atop \theta\in\overline\Theta}
\varphi_n(\theta)>\varphi_n(\theta_0).
\end{equation}
This proves that minimum point $\widehat\theta_n$ of
$\varphi_n(\theta)$ for $\theta\in\overline\Theta$ must be in the
ball $\{\|\theta-\theta_0\|\le \varepsilon\}$. Hence the convergence $\widehat\theta_n\rightarrow\theta_0$.

To complete the proof, it is sufficient to prove the claim, i.e. there is a compact set $\overline{\Theta}\subset\Theta$ such that for large $n$,
\begin{equation}\nonumber
\inf_{\theta\in\overline\Theta^c}
\varphi_n(\theta)>\varphi_n(\theta_0).
\end{equation}
Suppose not. Then there exists a sequence $\{\theta_l, l=1,2,\ldots\}$ tending to the boundary $\partial\Theta$ of $\Theta$ such that $\lim_{l\rightarrow\infty}\varphi_n(\theta_l)\leq\varphi_n(\theta_0)$.
Under this situation, we only need to consider the following two cases.

The first is that $\{\theta_l\}$ has a convergent sub-sequence, i.e. $\theta_{l_k}\rightarrow\theta\in\partial\Theta$, as $k\rightarrow\infty$, then it follows that $$0\leq\varphi(\theta)=\lim_{n\rightarrow\infty}\lim_{k\rightarrow\infty}\varphi_n(\theta_{l_k})
\leq\lim_{n\rightarrow\infty}\varphi_n(\theta_0)=\varphi(\theta_0)=0,$$
hence $\varphi(\theta)=0$.
By a similar technique used in the proof of Proposition \ref{pro2}, we may get $\theta=\theta_0$, a contradiction.

The second is that $||\theta_l||=\big(\sum_{i=1}^ka_{il}^2+\sum_{i=1}^{k-1}m_{il}^2\big)^{1/2}\rightarrow\infty$. Then we immediately know there exists $a_{il}$ such that $a_{il}\rightarrow\infty$, as $l\rightarrow\infty$. Without loss of generality, suppose that
\begin{eqnarray*}
\begin{cases}
a_{il}\rightarrow\infty                &(1\leq i\leq k_1),\\
\sum_{i=1}^{k_1}m_{il}\rightarrow m_{0}&\\
a_{il}\rightarrow a_{i}<\infty         & (k_1+1\leq i\leq k),\\
m_{il}\rightarrow m_{i}                & (k_1+1\leq i\leq k-1).
\end{cases}
\end{eqnarray*}
We have then
$$0\leq\lim_{n\rightarrow\infty}\lim_{l\rightarrow\infty}\varphi_n(\theta_{l})
\leq\lim_{n\rightarrow\infty}\varphi_n(\theta_0)=\varphi(\theta_0)=0,$$
and thus
\begin{eqnarray*}
\lim_{n\rightarrow\infty}\lim_{l\rightarrow\infty}\varphi_n(\theta_l)=
\sum_{j=1}^m\bigg(z_j-\frac{1-cm_0}{s_j}+c\sum_{i=k_1+1}^k\frac{a_{i} m_i}{1+a_{i}s_j}\bigg)^2=0.
\end{eqnarray*}
If $m_0=0$ then the problem is similar to the first case. Assume $m_0\neq0$.
Denote $\theta_0=(a'_{1},\ldots,a'_{k},m'_{1},\ldots,m'_{k-1})$, we have
\begin{eqnarray*}
\frac{m_0}{s_j}+\sum_{i=k_1+1}^k\frac{a_{i} m_{i}}{1+a_{i}s_j}=\sum_{i=1}^k\frac{a'_{i} m'_{i}}{1+a'_{i}s_j},
\end{eqnarray*}
for $j=1,\ldots,m$. Now look $s_j$ as a parameter $s$ and multiplying common factors
leads to
\begin{eqnarray*}
&&s\prod_{i_1=k_1+1}^k(1+a_{i_1}s)
\prod_{i_2=1}^k(1+a'_{i_2}s)\bigg(\frac{m_0}{s}+\sum_{i=k_1+1}^k\frac{a_{i} m_{i}}{1+a_{i}s}\bigg)\\
&=&s\prod_{i_1=k_1+1}^k(1+a_{i_1}s)
\prod_{i_2=1}^k(1+a'_{i_2}s)\bigg(\sum_{i=1}^k\frac{a'_{i} m'_{i}}{1+a'_{i}s}\bigg).
\end{eqnarray*}
These are polynomials of degree $2k-k_1\leq 2k-1$; they coincide at  $m\ge 2k-1$
different points $s=s_j$;   they are then equal.
Comparing their constant terms comes into conflict.

The proof is then complete.

\subsection{ Proof of Theorem \ref{th3}}\label{th3proof}
The proof of this theorem is similar to the proof of Theorem \ref{th2}. We only present the following proposition.

\begin{proposition}\label{pro3}
If $u_1,\ldots,u_m$ are distinct and $m\geq q$, then $\varphi(\theta)=0$ for the continues model has a unique solution $\theta_0$ on $\Theta$.
\end{proposition}

\begin{proof}
Suppose there is a $\theta=(\alpha_1,\ldots, \alpha_q)$ such that $\varphi(\theta)=0$.
Denote by $\theta_0=( \alpha'_1,\ldots, \alpha'_k)$ the true value of the parameter. We will show that $\theta=\theta_0$.

Define $p(t,\beta)=\beta_0+\beta_1t+,\cdots,+\beta_{q}t^{q},$ where $\beta=(\beta_1,\ldots,\beta_{q})$ and $\beta_0=1-\sum_{j=1}^qj!\beta_j$.
We have then
\begin{equation}\nonumber
   \int \frac{t}{1+ts_j}p(t,\theta^*)e^{-t}dt =  0  \quad (j=1,\ldots,m),
\end{equation}
where $\theta^*=\theta-\theta_0$ and $s_j=\us(u_j)$.

Suppose $p(t,\theta^*)=0$ has $q_0\ (\leq q)$ positive real roots $t_1< \ldots<t_{q_0}$, and denote $t_0=0, t_{q_0+1}=+\infty$, then $p(t,\theta^*)$ maintains the sign in each interval $(t_{i-1},t_i)\ (i=1,\ldots,q_0+1).$ By mean value theorem, we have
\begin{eqnarray*}
0=\int_0^{+\infty}\frac{t}{1+ts_j}p(t,\theta^*)e^{-t}dt=\sum_{i=1}^{q_0+1}\frac{\xi_i}{1+\xi_is_j}\int_{t_{i-1}}^{t_{i}}p(t,\theta^*)e^{-t}dt\quad (j=1,\ldots,m),
\end{eqnarray*}
where $\xi_i\in(t_{i-1},t_i)\ (i=1,\ldots, q_0+1)$.

Now look $s_j$ as a parameter $s$ and reduction to common factors
leads to
\[
0=\sum_{i=1}^{q_0+1}\prod_{l\neq i}(1+\xi_ls)\xi_i\int_{t_{i-1}}^{t_{i}}p(t,\theta^*)e^{-t}dt.
\]
The left hand side is a polynomial of degree $q_0-1\leq q-1$ (the coefficient of $ s^{q_0}= \prod_{j=1}^{q_0+1}\xi_j\int_0^\infty p(t,\theta^*)e^{-t}dt=0$); the equation has $m\ge q$ different roots $s=s_j$; the polynomial is then zero. Let $s=-1/\xi_i\ (i=1,\ldots,q_0+1)$, we get
$$\int_{t_{i-1}}^{t_i}p(t,\theta^*)e^{-t}dt=0\quad (i=1,\ldots,p_0+1),$$
which is followed by $p(t,\theta^*)=0$, and thus $\theta^*=0$.
\end{proof}

\end{document}